\documentclass[11pt]{article}

\usepackage{epsfig}
\usepackage{amsmath}
\usepackage{amssymb}
\usepackage{amsthm}
\usepackage{graphics}
\usepackage{psfrag}

\newtheorem{theorem}{Theorem}
\newtheorem{lemma}{Lemma}

\newcommand{\old}[1]{{}}

\newcommand{\RR}{\mathbb{R}}

\def\A{\mathcal A}

\def\x{\mathbf x}
\def\y{\mathbf y}

\title{On convexification of polygons by pops}

\author{%
Adrian Dumitrescu\thanks{Department of Computer Science,
University of Wisconsin--Milwaukee,
WI 53201-0784, USA\@. Email: \texttt{ad@cs.uwm.edu}.
Supported in part by NSF CAREER grant CCF-0444188.}
\and
Evan Hilscher
\thanks{Department of Computer Science,
University of Wisconsin--Milwaukee,
WI 53201-0784, USA\@. Email: \texttt{hilscher@uwm.edu}.
}}

\begin{document}

\maketitle

\begin{abstract}
Given a polygon $P$ in the plane, a {\em pop} operation is the reflection of
a vertex with respect to the line through its adjacent vertices. 
We define a family of alternating polygons, and show that any polygon
from this family cannot be convexified by pop operations. 
This family contains simple, as well as non-simple (i.e.,
self-intersecting) polygons, as desired. We thereby answer in the
negative an open problem posed by Demaine and O'Rourke~\cite[Open
Problem 5.3]{DO07}. 
\end{abstract}

\medskip
\hspace{0.15in}
\textbf{\small Keywords}:
Polygon convexification, edge-length preserving transformation,
pop operation.

\section{Introduction}

Consider a polygon $P=\{p_1,\ldots,p_n\}$ in the plane, that could
be simple or self-intersecting. 
A {\em pop} operation is the reflection of a vertex, say $p_i$, with respect
to the line through its adjacent vertices $p_{i-1}$ and  $p_{i+1}$
(as usual indexes are taken modulo $n$, i.e., $p_{n+1}=p_1$)~\cite{Ba03}. 
Observe that for the operation to be well-defined we need that 
$p_{i-1}$ and $p_{i+1}$ are distinct. This operation belongs to the
larger class of edge-length preserving transformations, when applied to
polygons~\cite{Ba03,Ro91,RW94,Sa73,SB92}. 
It seems to have been used for the first time by Millet~\cite{Mi94}.
If instead of reflecting $p_i$ with respect to the line through its
adjacent vertices $p_{i-1}$ and $p_{i+1}$, 
the reflection is executed with respect to the midpoint of 
$p_{i-1}$ and $p_{i+1}$, the operation is called a {\em popturn}; see
\cite{ABB+07,Ba03}. Observe that both the pop and the popturn are
single-vertex operations. 

Each is an instance of a ``flip'', defined informally, 
which has been studied at length. The most common variant
of flip is the {\em pocket flip} (or just {\em flip}), first
considered by Erd{\"o}s~\cite{Er35}. Another variant is the {\em
flipturn}, first considered by Kazarinoff, and later by Joss and Shannon; 
see~\cite{DGO+06b,Gr95} for an account of their results. 
In contrast with pops and popturns, both the flip and the flipturn may
involve multiple vertices. The inverse of a pocket flip, called {\em
deflation}, has been also considered~\cite{DDF+08,FHM+00}.
We briefly describe pocket flips and pocket flipturns next.

Assume that we deal with simple polygons in this paragraph.
A \emph{pocket} is a region exterior to the polygon but interior to
its convex hull, bounded by a subchain of the polygon edges and the
pocket \emph{lid}, the edge of the convex hull connecting the
endpoints of that subchain; see e.g., \cite[p.~74]{DO07}. Observe that
any non-convex polygon has at least one pocket. 
A {\em flip} of a pocket consists of reflecting the
pocket about the line through the pocket lid. Instead, a {\em flipturn}
of a pocket consists of reflecting the pocket about the midpoint
of the pocket lid. Observe that if $P$ is simple and non-convex, 
the polygons resulting after a pocket flip, or a pocket flipturn are
again simple.  It is known that within both of these variants,
convexification can be achieved. More precisely:  
given a simple polygon, it can be convexified by a finite sequence of 
pocket flips~\cite{DGO+06b,Gr95,GZ01,KB61,Na39,Re57,To99,W93,Yu57}.
Similarly, it can be convexified by a finite sequence of 
pocket flipturns~\cite{Gr95}. Moreover, the first result continues to hold
for self-intersecting polygons, under broad assumptions,
see~\cite{DGO+06b}. While the convexifying sequence can be arbitrarily
long for pocket flips (i.e., irrespective of $n$, the number of vertices), 
a quadratic number of operations always suffices in the case of
flipturns~\cite{ABC+00,ACD+02,Bi06}. 
There is an extensive bibliography pertaining to these 
subjects~\cite{ABC+00,ADE+01,ACD+02,ABB+07,Ba03,Bi06,DGO+06b,DO07,
Er35,Gr95,GZ01,Ka81,KB61,Na39,Re57,To99,W93,Yu57}.
See also ~\cite{Ba03,Ro91,RW94,Sa73,SB92} for more results
on edge-length preserving transformations and chord stretching.

In this paper we focus on pop operations. 
Thurston gave an example of a simple polygon that becomes
self-intersecting with any pop, see~\cite[p.81]{DO07}.
Ballinger and Thurston showed (according to~\cite[p.~81]{DO07})
that almost any simple polygon can be convexified by pops if
self-intersection is permitted; however no proof has been published. 
As Ballinger writes in his thesis~\cite{Ba03}, ``pops are very natural
transformations to consider, but the analysis of polygon
convexification by pops seems very tricky''. 
It has remained an open problem whether there exist polygons that
cannot be convexified by pops~\cite[Open Problem 5.3]{DO07}. 
We show here that such polygons do indeed exist, from both classes,
simple or self-intersecting, thereby answering the above open problem
in its full generality. 

In Section~\ref{sec:alt}, for every even $n \geq 6$, we define a family
$\A_n$ of {\em alternating polygons}, and show that any polygon from this
family cannot be convexified by pop operations. This family contains
simple, as well as non-simple (i.e., self-intersecting) polygons, as desired. 
It is interesting that this family is closed under pop operations: any
pop operation applied to a polygon in $\A_n$, at any vertex,
yields a polygon in $\A_n$ .

\section{Alternating polygons} \label{sec:alt}

Recall that in order for the pop operation on a vertex $p_i$ be well defined, 
its neighbors, $p_{i-1}$ and $p_{i+1}$ need to be distinct, so that 
the reflection line through them is unique, hence the reflection of
$p_i$ is also unique. A condition on the edge lengths of the polygon
that guarantees this is that no two edges have the same length; 
such a polygon is called {\em scalene}~\cite[p.~24]{Ba03}. 
A weaker condition that suffices is that no two consecutive edges have
the same length; we call such polygons {\em weakly scalene}. 
Our family of polygons $\A_n$ we define below consists of  weakly
scalene polygons.  

If $p_{i-1}$ and $p_{i+1}$ coincide, $p_i$ is called a {\em hairpin
vertex}~\cite{ABB+07}. Popping a hairpin vertex is undefined because
there are an infinite number of reflection lines through $p_{i-1}$ and
$p_{i+1}$. Our family of polygons is specifically designed to avoid
any occurrence of hairpin vertices. See~\cite{ABB+07} for a possible
adaptation of pops to hairpin vertices. 

Let $n$ be even. Fix a coordinate system in the plane.
We say that a polygon $P=\{p_1,p_2,\ldots,p_n\}$ with $n$ distinct
vertices is {\em alternating} if its vertices lie
alternately on the two axes: say, the vertices with {\em odd} indexes
on the $x$-axis, and the vertices with {\em even} indexes
on the $y$-axis. See Fig.~\ref{f1} for an illustration.
\begin{figure}[htbp]
\centerline{\epsfxsize=6.3in \epsffile{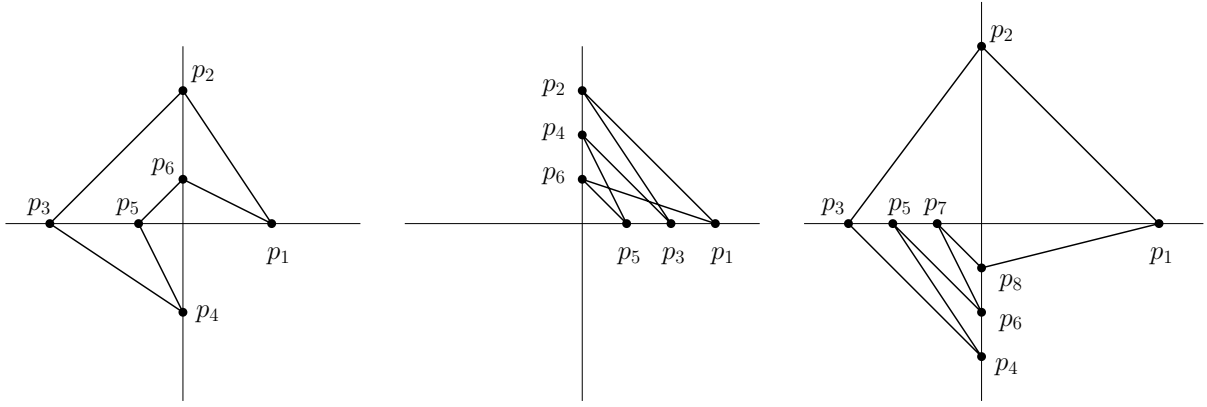}}
\caption{\small Alternating polygons with $6$ and $8$ vertices:
$A((2,3,1),(3,2,1),(+1,+1,-1,-1,-1,+1))$,
$A((3,2,1),(3,2,1),(+1,+1,+1,+1,+1,+1))$, and 
$A((4,3,2,1),(4,3,2,1),(+1,+1,-1,-1,-1,-1,-1,-1))$. 
The one in the middle is self-intersecting.} 
\label{f1}
\end{figure}

Let $n=2k$. Let $\x=(x_1,x_2,\cdots,x_k)$, and $\y=(y_1,y_2,\cdots,y_k)$ 
be two vectors in the positive orthant of $\RR^k$, 
each having distinct nonzero coordinates, that is:
\begin{align*} \label{E1}
i \in \{1,2,\ldots,k\} &\quad\Rightarrow\quad x_i>0 {\rm \ and \ } y_i>0, \\
i,j \in \{1,2,\ldots,k\}  {\rm \ and \ } i \neq j &\quad\Rightarrow\quad 
x_i \neq x_j {\rm \ and \ } y_i \neq y_j. 
\end{align*} 
Let $\sigma=(\sigma_1,\sigma_2,\cdots,\sigma_{2k}) \in \{-1,+1\}^{2k}$
be a binary sign vector. 
Consider the alternating polygon $A(\x,\y,\sigma)=\{p_1,p_2,\ldots,p_{2k}\}$, where 
\begin{itemize}
\item $p_{2i+1}=(\sigma_{2i+1} \cdot x_{i+1},0)$, for $i=0,\ldots,k-1$. 
\item $p_{2i}=(0,\sigma_{2i} \cdot y_i)$, for $i=1,\ldots,k$. 
\end{itemize}
Let $\A_n$ ($\equiv \A_{2k}$) be the family of all alternating
polygons $A(\x,\y,\sigma)$ defined 
as above. First note that $\A_n$ contains both simple, as well as
non-simple (i.e., self-intersecting) polygons.  

Indeed, consider the polygon $P_1$ described next.
Let $x_1=y_1=k$, and $x_i=y_i=k-i+1$, for $i=2,\ldots,k$. 
Let $\sigma=(+1,+1,-1,\ldots,-1)$.
It is easy to see that $P_1 \in \A_n$ is a simple polygon. An
example is shown in Fig.~\ref{f1}~(right).  

Consider now the polygon $P_2$ described as follows.
Let $x_i=y_i=k-i+1$, for $i=1,\ldots,k$. Let $\sigma=(+1,\ldots,+1)$.
It is easy to see that $P_2 \in \A_n$ is a self-intersecting polygon. An
example is shown in Fig.~\ref{f1}~(middle).  

A sequence of pops executed on an alternating simple polygon with 
$6$ vertices appears in Fig.~\ref{f2}.  
\begin{figure}[t]
\centerline{\epsfxsize=5in \epsffile{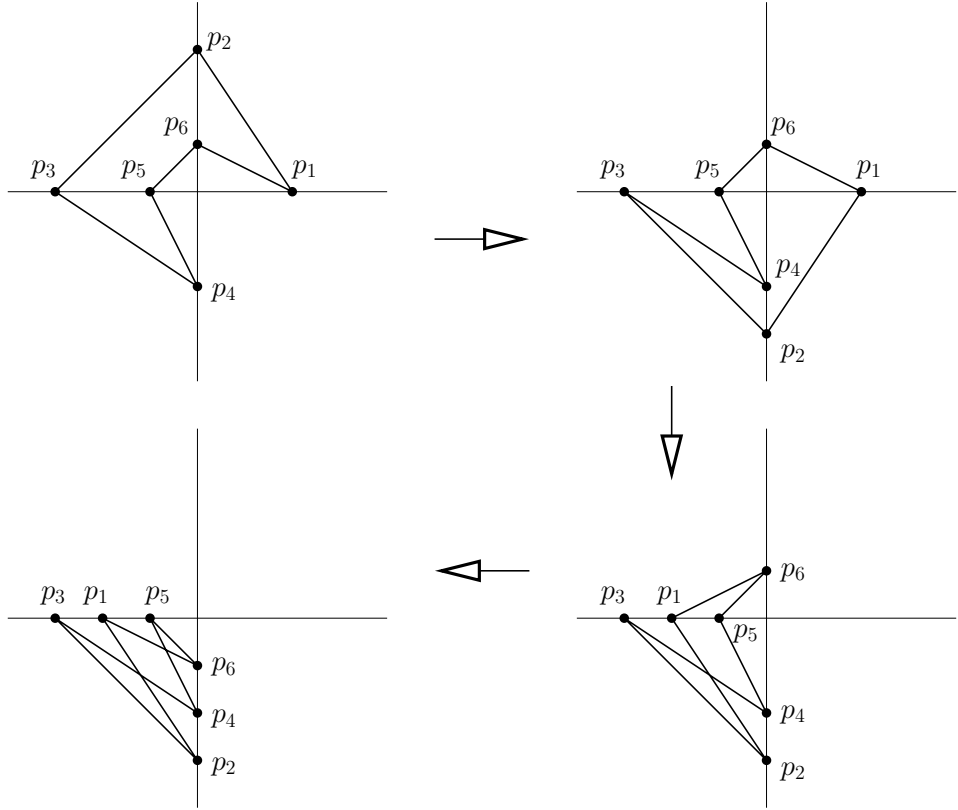}}
\caption{\small Sequence of three pops executed on vertices $p_2$, $p_1$, and
$p_6$ of $\{p_1,p_2,p_3,p_4,p_5,p_6\}$. The corresponding polygon sequence is 
$A((2,3,1),(3,2,1),(+1,+1,-1,-1,-1,+1)) \Rightarrow$
$A((2,3,1),(3,2,1),(+1,-1,-1,-1,-1,+1)) \Rightarrow$
$A((2,3,1),(3,2,1),(-1,-1,-1,-1,-1,+1)) \Rightarrow$
$A((2,3,1),(3,2,1),(-1,-1,-1,-1,-1,-1))$.}
\label{f2}
\end{figure}
A key fact regarding alternating polygons is the following:

\begin{lemma} \label{L1}
If $P \in A_{2k}$ is convex, then $k \leq 2$.
\end{lemma}
\begin{proof}
Since $P$ is convex, it intersects each of the coordinate axes in at
most two points, unless it is tangent to one of the coordinate axes,
and there are three consecutive collinear vertices on that axis. 
However this latter possibility would contradict the alternating
property of $P$. So the only alternative is the former, in which case
we have $k \leq 2$. Observe that the given inequality on $k$ cannot
be improved. 
\end{proof}

The following properties are easy to verify:
\begin{enumerate}
\item $A(\x,\y,\sigma)$ has $2k$ distinct vertices.
\item $A(\x,\y,\sigma)$ is weakly scalene.
\item The pop operation applied to the vertex $p_i$ of $A(\x,\y,\sigma)$,
($1 \leq i \leq 2k$), yields $A(\x,\y,\sigma')$, where $\sigma'$ differs
from $\sigma$ only in the $i$th bit. That is, the absolute value of 
the non-zero coordinate of $p_i$ remains the same, with the point
switching to its mirror image with respect to the origin of the axes.
In particular, this implies that the family $\A_{2k}$ is closed with
respect to pop operations. 
\item Let $\x,\y$ be fixed, with the above properties,
and $\sigma, \sigma'$ be two sign vectors. 
Consider $P=A(\x,\y,\sigma)$, and $P'=A(\x,\y,\sigma')$.
Then $P'$ can be obtained from $P$ by executing at most $n$ pops, via: 
For $i=1$ to $n$ do: if $\sigma_i \neq \sigma'_i$, then pop $p_i$ to
$p'_i$.  
\end{enumerate}

We are now ready to prove our main result:
\begin{theorem}\label{T1}
Let $n=2k$, where $k \geq 3$. 
Any polygon in the family $\A_n$ is non-convexifiable by pop operations.
\end{theorem}
\begin{proof}
Consider a polygon $P \in A_{2k}$.
(We can choose $P$ simple, or self-intersecting, as desired.)
By Lemma~\ref{L1}, $P$ is not convex. Apply any finite sequence of pop
operations. By the second property (2.) above, the resulting polygon also
belongs to $A_{2k}$, and is therefore not convex. 
\end{proof}

\section{Conclusion} \label{sec:conclusion}

We have shown that there exists a family of polygons that cannot be
convexified by a finite sequence of pops. However, there exist many
polygons that can be convexified in this way. We conclude with two questions: 
\begin{enumerate}
\item What is the computational complexity of deciding whether a
given (simple or self-crossing) polygon can be convexfied by a 
finite sequence of pops?
\item How hard is it to find a shortest sequence of pops that
convexifies a given polygon (assuming it is convexifiable in this
way)? Do good approximation algorithms exist for this problem?
\end{enumerate}

\end{document}